\documentclass{article}

\usepackage[T1]{fontenc}
\usepackage{ifthen}

\usepackage{xcolor}
\definecolor{defblue}{rgb}{0.121,0.47,0.705}
\definecolor{linkblue}{rgb}{0.098,0.098,0.4392}
\usepackage[colorlinks=true, 
    linkcolor=linkblue, 
    anchorcolor=linkblue,
    citecolor=linkblue, 
    filecolor=linkblue,
    menucolor=linkblue,
    urlcolor=linkblue, 
    bookmarks=true, 
    bookmarksopen=true,
    bookmarksopenlevel=2, 
    bookmarksnumbered=true, 
    hyperindex=true,
    plainpages=false, 
    pdfpagelabels=true
]{hyperref}

\usepackage{amsmath}
\usepackage{fullpage}
\usepackage{amsthm}
\usepackage{amssymb}
\usepackage{xspace}
\usepackage{enumitem}
\usepackage{comment}
\usepackage{authblk}
\usepackage{graphicx}

\newcommand{\ConjOUT}[1]{}

\usepackage{caption}
\usepackage{subcaption}

\usepackage{thmtools,apptools,thm-restate}
\newcommand{\restateref}[1]{\IfAppendix{\hyperref[#1]{$\star$}}{\hyperref[#1*]{$\star$}}}
\usepackage{url}

\usepackage{todonotes}

\usepackage[capitalise,noabbrev,nameinlink]{cleveref}

\newtheorem{theorem}{Theorem}
\newtheorem{lemma}[theorem]{Lemma}

\newtheorem{observation}{Observation}

\newtheorem{corollary}[theorem]{Corollary}
\newtheorem{claim}{Claim}

\newtheorem{myclaim}{Claim}{\bfseries}{\itshape}
\crefname{myclaim}{Claim}{Claims}

\crefname{conjecture}{Conjecture}{Conjectures}
\newtheorem{question}{Question}

\crefname{question}{Question}{Questions}
\Crefname{figure}{Figure}{Figure}

\newcommand{\ER}{\ensuremath{\exists\mathbb{R}}\xspace}

\newcommand{\NP}{\ensuremath{\text{NP}}\xspace}

\newcommand{\R}{\mathbb{R}}
\newcommand{\N}{\mathbb{N}}

\newcommand{\A}{\ensuremath{\mathcal{A}}\xspace}
\renewcommand{\S}{\ensuremath{\mathcal{S}}\xspace}

\DeclareMathOperator{\qn}{qn}

\newcommand{\stretchability}{\textsc{Stretchability}\xspace}
\newcommand{\segmentStretchability}{\textsc{Pseudo-Segment Stretchability}\xspace}

\newcommand{\geo}{\textsc{Geometric Thickness}\xspace}

\newcommand{\thechromaticnumber}{30}
\newcommand{\themaxdegree}{72}
\newcommand{\Sum} [2] {\the\numexpr #1 + #2 \relax}
\newcommand{\Prod} [2] {\the\numexpr #1 * #2 \relax}
\newcommand{\chromaticnumber}{\ensuremath{\thechromaticnumber}\xspace}
\newcommand{\chromaticnumberplusone}{\ensuremath{\Sum{\thechromaticnumber}{1}}\xspace}
\newcommand{\thekplanarnumber}{\ensuremath{\Prod{\themaxdegree}{\Sum{\Prod{\thechromaticnumber}{2}}{1}}\xspace}}
\newcommand{\kplanarnumber}{\ensuremath{\thekplanarnumber}\xspace}
\newcommand{\sge}{\textsc{Simultaneous Graph Embedding}\xspace}

\title{Geometric Thickness of Multigraphs is \ER-Complete}

\author[1]{Henry F\"{o}rster}
\author[2]{Philipp~Kindermann}
\author[3]{Tillmann~Miltzow}
\author[4]{Irene Parada}
\author[5]{Soeren Terziadis}
\author[6]{Birgit Vogtenhuber}

\affil[1]{Wilhelm-Schickard Institut für Informatik, Universit\"{a}t T\"{u}bingen}
\affil[2]{FB IV - Computer Science, Trier University}
\affil[3]{Department of Information and Computing Sciences, Utrecht University}
\affil[4]{Departament de Matemàtiques, Universitat Politècnica de Catalunya}
\affil[5]{Institute of Logic and Computation, Technische Universität Wien}
\affil[6]{Institute of Software Technology, Graz University of Technology}

\begin{document}

\maketitle

\abstract{
    We say that a (multi)graph $G = (V,E)$ has geometric thickness $t$ if there exists
    a straight-line drawing $\varphi : V \rightarrow \R^2$ and a $t$-coloring of its edges where 
    no two edges sharing a point in their relative interior have the same color.
    The \geo problem asks whether a given multigraph has geometric thickness at most $t$. 
    This problem was shown to be NP-hard for $t=2$ [Durocher, Gethner, and Mondal, CG 2016]. 
    In this paper, we settle the computational complexity of \geo by showing that it is \ER-complete already for thickness \chromaticnumber. 
    Moreover, our reduction shows that the problem is \ER-complete for \kplanarnumber-planar graphs, where a graph is $k$-planar if it admits a topological drawing with at most $k$ crossings per edge. 
    In the course of our paper we answer previous questions on geometric thickness and on other related problems, in particular that
    simultaneous graph embeddings of \chromaticnumberplusone edge-disjoint graphs and pseudo-segment stretchability with chromatic number~\chromaticnumber are \ER-complete. 
    }

\section{Introduction}

This paper shows that \geo is \ER-complete, for multigraphs and geometric thickness at least \chromaticnumber.
We start with a tangible example, historic background, and motivation. 
Then we state our results formally,
followed by an in-depth discussion.
We conclude this section by giving an introduction to the existential theory of the reals and a sketch of the \ER-completeness proof of segment stretchability due to Schaefer~\cite{S21b}.

\begin{figure}
    \centering
    \includegraphics
    [width=.9\linewidth]
    {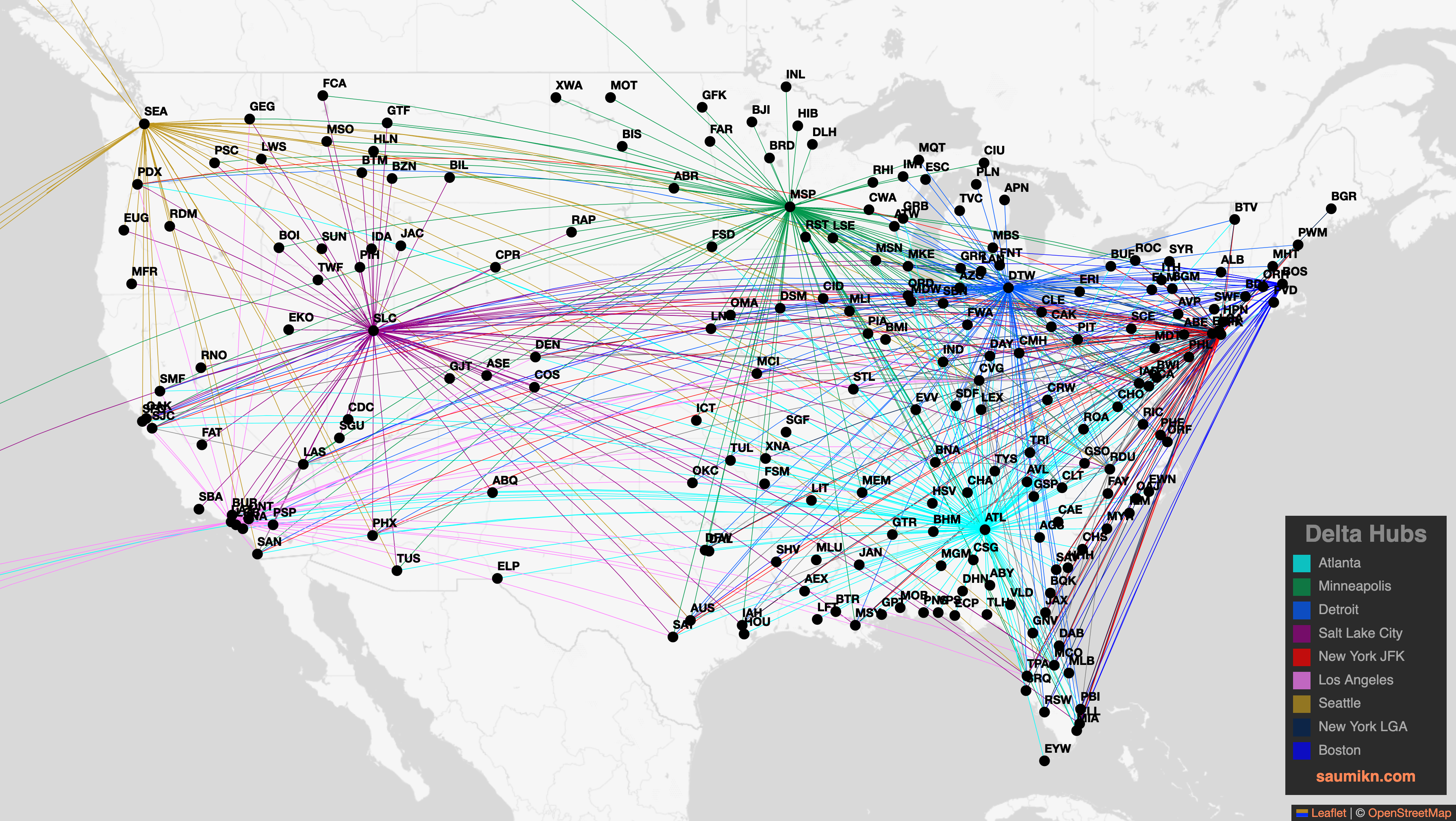}
    \caption{
    Illustration of Delta airline connections within the continental United States. Connections from different hubs are colored differently.
    Source: \href{https://saumikn.com/blog/airlinemaps/}{https://saumikn.com/blog/airlinemaps/}, used with permission.
    }
    \label{fig:practical_examples}
\end{figure}

\paragraph{Tangible Examples.} Assume you want to visualize the flight network of the continental United States of America, see \Cref{fig:practical_examples}.
A natural way to do so is to find a drawing with cities as points and connections by straight lines.
Now, if the edges are not crossing this often gives an insightful illustration revealing patterns that are not easily determined by looking at the adjacency matrix.
Unfortunately, often networks, like the airline connection network, have many edges and it is not possible to obtain a planar drawing.
In this case, we can color the edges in the drawing. 
In case that there is no monochromatic crossing and few colors, it is still relatively easy for humans to process the graph in a way that we can extract useful information from it.

In this work, we study the algorithmic question of finding such a drawing together with a correct coloring.

\paragraph{Historic Background}
The \emph{thickness} of a graph $G$ is the minimum number of planar subgraphs whose union is $G$. 
It is an old concept; it was formally introduced by Tutte~\cite{Tutte1963thickness} in 1963, 
but the concept of biplanarity (i.e., geometric thickness 2) had already appeared before, 
most relevantly, in connection with two open problems: 
First, Ringel's Earth–Moon problem on the chromatic number of biplanar graphs~\cite{Ringel1959coloring} 
and second, a question by Selfridge, formulated by Harary, asking whether $K_9$ is biplanar~\cite{Harary1961thickness,BattleHK1962K9,Tutte1963K9}. 
In 1983, Mansfield~\cite{Mansfield1983thickness} showed that deciding whether a graph is biplanar is NP-complete. 

In this article we study the geometric or straight-line version of thickness, which requires that all planar subgraphs are embedded simultaneously with straight-line edges. 
More precisely, a multigraph (In a multigraph, the edge set $E$ may contain multiple edges connecting the same pair of vertices.) 
$G = (V,E)$ has \emph{geometric thickness} $t$ if there exists a straight-line drawing $\varphi : V \rightarrow \R^2$ of $G$ and a $t$-coloring of all the edges such that no two edges of the same color share a point other than a common endpoint. \Cref{fig:Kn_examples} shows an illustration. 
Note that by definition, two edges connecting the same two endpoints must be assigned distinct colors in the $t$-coloring.

\begin{figure}[t]
    \centering
    \includegraphics[width=0.9\textwidth,height=3.2cm]{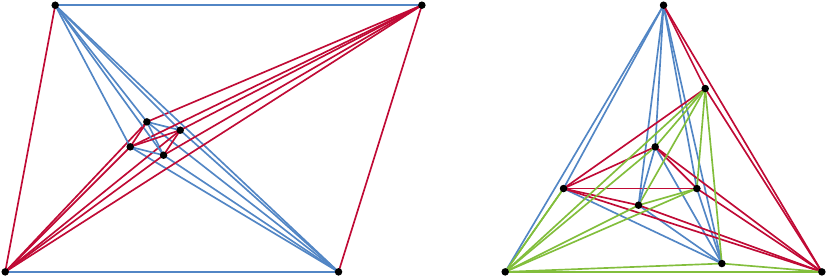}
    \caption{The (geometric) thickness of $K_8$ is 2 and of $K_9$ is 3.
    }
    \label{fig:Kn_examples}
\end{figure}

The concept of geometric thickness was introduced by Dillencourt, Eppstein, and Hirschberg~\cite{DEH00}, who studied the geometric thickness of complete and complete bipartite graphs. 
They already asked about the computational complexity of \geo, the problem of deciding whether given a graph $G$ and a value $t$, $G$ has geometric thickness at most $t$.
Durocher, Gethner, and Mondal~\cite{durocher2016thickness} partially answered this question by showing that \geo is  NP-hard 
even for geometric thickness~2.

\paragraph{Relation Between Geometric Thickness and Other Parameters}
~\Cref{fig:Kn_examples} shows straight-line drawings of $K_8$ and $K_9$ decomposed into two and three plane subgraphs, respectively; these bounds are tight for both the thickness and the geometric thickness. 
For $n>10$ the thickness of $K_n$ is $\lfloor \frac{n+2}{6} \rfloor$~\cite{AlekseevG1976thicknessKn,BeinekeH1965thicknessKn} while the geometric thickness, for which no tight bound is known, is lower bounded by $\frac{n}{5.646}$~\cite{DEH00}.

In general, the geometric thickness of a graph is not bounded by any function of its thickness. 
Eppstein~\cite{Eppstein2004thickness} proved that for every $t$, there exists a graph with thickness 3 and geometric thickness at least $t$.

The geometric thickness of graphs has also been studied in connection with the degree. 
The geometric thickness of graphs with maximum degree four is two~\cite{DuncanEK04lowdeg}. 
This result does not generalize: based on counting techniques, Bar\'{a}t, Matou\v{s}ek and Wood~\cite{BaratMW2006largethickness} showed that there exist bounded degree graphs with arbitrarily large geometric thickness. 
A graph is \emph{2-degenerate} if every subgraph contains a vertex of degree at most 2. It was recently shown~\cite{JainRRS2023twodegenerate} that 2-degenerate graphs have geometric thickness at most 4 and some of them have geometric thickness at least 3.

Regarding the density, Dujmovi\'c and Wood~\cite{DujmovicW2018antithickness} showed that every graph with $n\geq3$ vertices and geometric thickness $k\le n/2$ has at most $k(3n-k-5)$ edges and there are arbitrarily large graphs with geometric thickness $k$ and $k(3n-4k-3)$ edges.

\paragraph{Motivation}
The \geo problem combines two computationally hard problems: 
splitting the edges into color classes and positioning the vertices. 
The first problem, where we are given the straight-line drawing and the goal is to decompose it into the minimum number of plane subgraphs, corresponds to a graph coloring problem for the corresponding segment intersection graph of the drawing. 
The 2022 Computational Geometry Challenge focused on this problem~\cite{CG-SHOP22}.

The second problem, when we are given the color classes and the goal is to position the vertices such that no two edges of the same color class intersect in their relative interior, is the \sge problem, which we will formally define later. Thus, \geo and \sge differ in whether or not the partition of the edges is already given. As the partition of edges is combinatorial property, it is not the part that makes \geo \ER-hard. However, existing \ER-hardness results for \sge~\cite{S21b} do not immediately transfer to \geo as it may be possible that for some partitions it is in fact easy to find a corresponding embedding.
Both \sge and \geo connect to practical applications.

In network visualization, in particular of infrastructure, social, and transportation networks, one often has to deal with
intersecting 
systems of connections belonging to different subnetworks. 
To represent them simultaneously, different visual variables such as colors are used to indicate edge classes. 
Drawing the edges with straight-line segments and removing/minimizing same-class crossings is often desirable for readability.

In some settings, the vertex positions can be freely chosen, while the edges classes are given.
An example would be visualizing a system of communication channels (phone, email, messenger services) between a set of persons. 
This case corresponds to the \sge problem. 
In other settings, both the positions and the edge classes can be freely chosen, which corresponds to the \geo problem. 
It appears in applications such as VLSI design, where a circuit using uninsulated wires requires crossing wires to be placed on different layers~\cite{MutzelOS1998thickness}.

Since its introduction as a natural measure of approximate planarity~\cite{DEH00}, geometric thickness keeps receiving attention in Computational Geometry. 
However, some fundamental questions remain open, including determining the geometric thickness of $K_n$ and the complexity of \geo. 
Moreover, the computational methods currently available are not able to provide straight-line drawings of low geometric thickness for graphs with more than a few vertices. 
 
Our results show that, even for small constant values of geometric thickness, computing such drawings is, under widely-believed computational complexity assumptions, harder than any NP-complete problem.

\subsection{Results and Related Problems}
In this section we state all our results, starting with our main result. 
To this end, we first need to introduce 
the complexity class \ER.
The class \emph{\ER} 
    can be defined as the set of problems that are at most as difficult as finding a real root of a multivariate polynomial with 
    integer coefficients. 
    A problem in \ER is \emph{\ER-complete} if it is as difficult as this problem.
    We give a more detailed introduction to \ER in \Cref{sub:etr}. 
We are now ready to state our main result.

\begin{theorem}
    \label{thm:para-ER}
   \geo is \ER-complete for multigraphs already for any fixed thickness $t \geq \chromaticnumber$. 
\end{theorem}

Assuming $\ER \neq$ NP,
\Cref{thm:para-ER} shows that \geo is even more difficult than any problem in NP. 
This implies that even SAT and MILP solvers should not be able to solve \geo in full generality.
Even more, while a planar graph on $n$ vertices can be drawn on an $n\times n$ grid,
there are graphs with geometric thickness $\geq \chromaticnumber$
that will need 
more than exponentially large integer coordinates for any drawing.

We prove \Cref{thm:para-ER} in \Cref{sec:geothickness} via a reduction from the problem \segmentStretchability, 
which we discuss next.

\paragraph{Pseudo-Segment Stretchability} 
Schaefer showed that the problem \segmentStretchability (see below for a definition) is \ER-complete~\cite{S21b}.
We closely inspected the proof by Schaefer (see also \cref{sec:reduction_description}) and observed some small extra
properties that we use for our bound in \Cref{thm:para-ER} (see also \cref{sec:chromaticnumber}).
To state our additional observations, we first introduce the corresponding definitions.

A \emph{pseudo-segment arrangement}
in $\R^2$ is an arrangement of 
 Jordan arcs
 such that any two arcs intersect at most once.
 (A \emph{Jordan arc} is a continuous   curve described by an injective function from a closed interval to the plane.)
A pseudo-segment arrangement \A is called \emph{stretchable} if there exists a segment arrangement~\S 
such that \A and \S are isomorphic.
The curves 
can be encoded using a planar graph via planarization, as it completely determines the isomorphism type which 
describes the order in which each pseudo-segment intersects
all other pseudo-segments.
However, in contrast to (bi-infinite) pseudolines, this information is not enough.
For example the pseudo-segment arrangement in \Cref{fig:clouds} (left) cannot be stretched,
although there exists a segment arrangement with the same intersection pattern (right). 
Additionally, we require
the clockwise cyclic order around each intersection point to be maintained.

\begin{figure}[tbp]
    \centering
    \includegraphics[scale=0.95]{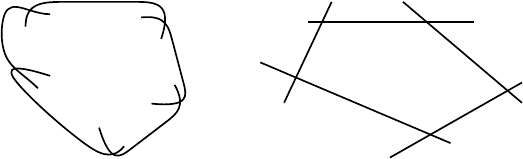}
    \caption{A stretchable (right) and a non-stretchable (left) pseudo-segment arrangement with the same intersection pattern but different cyclic order around the intersections.}
    \label{fig:clouds}
\end{figure}

\segmentStretchability asks whether a given pseudo-segment arrangement is stretchable.
Schaefer~\cite{S21b} showed that \segmentStretchability is  be \ER-hard, even 
if each pseudo-segment intersects at most $72$ other pseudo-segments.
We consider the intersection graph $G$ of all the
pseudo-segments, and we say a pseudo-segment arrangement has 
\emph{chromatic number} $\chi$
if and only if the underlying intersection graph $G$
has this chromatic number.
We prove the following corollary based on and strengthening a result of Schaefer~\cite{S21b} in \Cref{sec:chromaticnumber}. 

\begin{restatable}{corollary}{corColoring}
    \label{cor:coloring}
    \segmentStretchability is \ER-hard even for pseudo-segment arrangements 
    of chromatic number \chromaticnumber.
\end{restatable}

\paragraph{Simultaneous Graph Embedding}
A slight modification of our main result happens to also resolve a question by Schaefer about  simultaneous graph embeddings in the sunflower setting~\cite{S21b}.
Before stating the result, we need further definitions.
Given simple graphs $G_1,\ldots, G_k$ on the same vertex set $V$, we say $\varphi : V \rightarrow \R^2$ is a \emph{simultaneous} graph embedding, if the straight-line drawing of each $G_i$, $i = 1,\ldots, k$, on this vertex set is crossing-free.
We say $G_1,\ldots, G_k$ form a \emph{sunflower} if every edge is either in all graphs or in exactly one of the graphs.
An edge that is present in all graphs is \emph{in the center} of the sunflower or a \emph{public} edge.
Edges that are only in one graph are \emph{private} edges or \emph{belonging to a petal} of the sunflower. 
If all edges are private, we say that $G_1,\ldots, G_k$ form an \emph{empty} sunflower.
We define \sge as the algorithmic problem with $k$ graphs 
as input that asks if there exists a simultaneous graph embedding.

\begin{corollary}
    \label{thm:SGE}
    \sge is \ER-complete already with $k = \chromaticnumberplusone$ and 
    $G_1,\ldots,G_k$ forming an {empty} sunflower.
\end{corollary}

The proof can be found in \Cref{sec:SGE}.
In a series of papers 
it was shown that \sge is \ER-complete~\cite{CK15,EBGJPSS08,S21,S21b},
already for $k = 240$.
We contribute to this line of research by lowering the
bound to \chromaticnumberplusone and by restricting it to families of graphs
that form an empty sunflower. 
Lowering the number of graphs from $240$ to \chromaticnumberplusone is not so significant as \chromaticnumberplusone is still relatively high.
The  significance lies in the fact that we show \ER-hardness also for the case that the input graphs form an empty sunflower, which answers a question of Schaefer~\cite{S21b}. 

\paragraph{Relation Between Geometric Thickness and $k$-Planarity}
In a \mbox{\emph{$k$-planar}} drawing of a graph each  edge can be crossed at most $k$ times. 
A $k$-planar graph is a graph admitting such a drawing. 
For the class of $1$-planar graphs, $2$ is a tight upper bound for both the thickness and the geometric thickness. 
The upper bound for the thickness follows immediately from the degree of the intersection graph.
 
The upper bound of $2$ for the geometric thickness is non-trivial since not every 1-planar graph has a straight-line 1-planar drawing and 
was only recently shown~\cite{Brandenburg2021oneplanar}.
This in particular means that the (geometric) thickness of $1$-planar graphs can be decided in polynomial time. 
In contrast, we show at the end of \Cref{sec:geothickness} that determining the geometric thickness of \kplanarnumber-planar multigraphs is \ER-complete. 

\begin{corollary}
        \label{cor:kpla-ER}
   \geo is \ER-complete for $\kplanarnumber$-planar multigraphs. 
\end{corollary}

Brandenburg~\cite{Brandenburg2021oneplanar} explicitly asked about the geometric thickness of \mbox{$k$-planar} graphs. 
A combination of previous results~\cite{vidapat22} shows that \emph{simple} $k$-planar graphs have bounded geometric thickness:
A queue layout of a graph $G$ is a total vertex ordering and a partition of the edges into \emph{queues} such that no queue contains two nested edges $ad$, $bc$ such that $a < b < c < d$. The \emph{queue number} $\qn(G)$ of $G$ is the minimum number of queues in such a queue layout.
First, geometric thickness is known to be bounded by queue number: any graph $G$ with queue number $\qn(G)$ has geometric thickness $\theta(G) \le 8\qn(G) ^3$~\cite[Corollary 9]{DujmovicPW2004track}.
Second, it was recently shown that any $k$-planar graph $G$ has bounded queue number~\cite{DujmovicJMMUW2020queue,DujmovicMW2023kplanar}, more precisely $\qn(G) \le 2\cdot490^{k+2}$.
Together, this yields:

\begin{corollary}
    Any simple $k$-planar graph has geometric thickness at most \mbox{$64\cdot490^{3(k+2)}$}.
\end{corollary}

\subsection{Existential Theory of the Reals}
\label{sub:etr}

The complexity class \ER (pronounced as ``ER'') has gained a lot of interest in recent years.
It is defined via its canonical complete problem ETR (short for \emph{Existential Theory of the Reals}) and contains all problems that are polynomial-time many-one reducible to it.
In an ETR instance we are given an integer~$n$ and a sentence of the form
\[
    \exists X_1, \ldots, X_n \in \R :
    \varphi(X_1, \ldots, X_n),
\]
where~$\varphi$ is a well-formed and quantifier-free formula consisting of polynomial equations and inequalities in the variables and the logical operators $\{\land, \lor, \lnot\}$.
The goal is to decide whether this sentence is true.
As an example consider the formula $\varphi(X,Y) :\equiv X^2 + Y^2 \leq 1 \land Y^2 \geq 2X^2 - 1$.
Among (infinitely many) other solutions, $\varphi(0,0)$ evaluates to true, witnessing that this is a yes-instance of ETR.
We find it worth noting that one can define \ER with a real witness and a real verification algorithm, similar to the way that NP is defined. 
The difference is that the witness is allowed to have real number input and the verification runs on the real RAM instead of the word RAM~\cite{EvdHM20}. %
It is known that
\[
    \textrm{NP} \subseteq \ER \subseteq \textrm{PSPACE}
    \text{.}
\]
Here the first inclusion follows because a SAT instance can trivially be written as an equivalent ETR instance (in particular, note that  $X^2=X$ has real solutions $X=0$ and $X=1$).
The second inclusion is highly non-trivial and was first shown by Canny~\cite{Canny1988_PSPACE}.

Note that the complexity of problems involving real numbers was studied in various contexts.
To avoid confusion, let us make some remarks on the underlying machine model.
As already noted above, the underlying machine model for \ER (over which sentences need to be decided and where reductions are performed in) is the word RAM (or equivalently, a Turing machine) 
and not the real RAM~\cite{EvdHM20} or the Blum-Shub-Smale model~\cite{BSS89}.

The complexity class \ER gains its importance by numerous important algorithmic problems that have been shown to be complete for this class in recent years.
The name \ER was introduced by Schaefer in~\cite{S10} who also pointed out that several \NP-hardness reductions from the literature actually implied \ER-hardness.
For this reason, several important \ER-completeness results were obtained before the need for a dedicated complexity class became apparent.

Common features of \ER-complete problems are their continuous solution space and the nonlinear relations between their variables.
Important \ER-completeness results include the realizability of abstract order types~\cite{M88,S91} and geometric linkages~\cite{S13}, as well as the recognition of geometric segment~\cite{KM94,M14}, unit-disk~\cite{McDM10},    and ray intersection graphs~\cite{CFMTV18}.
More results appeared in the graph drawing community~\cite{S21,DKMR18,E19,LM20}, regarding polytopes~\cite{AKM23,APT15,DHM19,RG96,RGZ95,V23}.
Furthermore, other important areas are the study of Nash-equilibria~\cite{AOB23,BHM23,B16,CIJ14,EY08,H19,HS20,SS17}, machine learning~\cite{AKM21,BHJMW22,B22b,MII22} matrix factorization~\cite{L09,SS18,S16,shitov2017complexity,S23PSD}, or continuous constraint satisfaction problems~\cite{MS22}.
In computational geometry, we would like to mention the art gallery problem~\cite{AAM18,S22}, geometric packing~\cite{AMS20} and covering polygons with convex polygons~\cite{A22}.

Recently, the community started to pay more attention to higher levels of the ``real polynomial hierarchy'', which also capture several interesting algorithmic problems~\cite{BH21,BC09,DCLNOW21,DKMR18,JKM22,SS22}.

\subsection{\ER-hardness of \stretchability}\label{sec:reduction_description}
In this section, we review Schaefer's proof of \ER-hardness of \stretchability~\cite{S21b}. In \cref{sec:chromaticnumber}, we will show that a modified version of the pseudo-segment arrangement constructed by Schaefer can be colored with at most $30$ colors which gives rise to the values of $t$ and $k$ in \cref{thm:para-ER} and \cref{thm:SGE}, respectively. 
To understand where our modifications happen and why they do not impact the correctness of the reduction, we consider it necessary to recall how Schaefer's reduction works in detail.

In~\cite{S21b}, Schaefer reduces the so-called \textsc{Strict Inequality} problem to \stretchability. In \textsc{Strict Inequality}, the input, also called \emph{program}, is a set of strict inqualities of polynomials in the variables $X_1, \ldots, X_n$ and the problem is to decide whether there are real values for    $X_1, \ldots, X_n$ such that all inequalities are fulfilled. \textsc{Strict Inequality} has  been shown to be \ER-hard; see e.g. \cite{SS17}. 
In fact, it is sufficient to consider only restricted inputs for \textsc{Strict Inequality}. Richter-Gebert~\cite{Richter1995} showed that for any program $P$ for \textsc{Strict Inequality} there exists a program $P^*$ in \emph{Richter-Gebert normal form} (\textrm{RG-NF}) for \textsc{Strict Inequality} such that $P^*$ is a positive instance if and only if $P$ is a positive instance.

Aside from the to-be-determined variables $X_1, \ldots, X_n$, a program in Richter-Gebert normal form (\textrm{RG-NF}) uses a set of \emph{computed variables} ($V_1,\ldots,V_m)$ which are used to store computed results of the  polynomials that are to be compared. Moreover, there are only six types of polynomials, namely:
\begin{enumerate}
    \item\label{type:unit} $V_i=1$,
    \item\label{type:equality} $V_i=X_j$ with $X_j > 1$,
    \item\label{type:negation} $V_i=-V_j$ with $i > j$ and $V_j > 1$ or $V_j=1$,
    \item\label{type:negated-addition} $V_i=-V_j+V_k$ with $i > j,k$, $V_j < 0$ and $V_k > 1$ or $V_k=1$,
    \item\label{type:inversion} $V_i=(1/V_j)$ with $i > j$ and $V_j > 1$, and
    \item\label{type:invertedMultiplication} $V_i=(1/V_j)\cdot V_k$ with $i > j,k$, $0 < V_j < 1$ and $V_k > 1$.
\end{enumerate}
Note that the types of polynomials implicitly provide some strict inequalities of type $V_i < V_j$ for $1 \leq i < j \leq m$ (e.g. the negated addition, i.e., Type \ref{type:negated-addition}, requires $V_i$ to be greater than both $V_j$ and $V_k$). In addition, there may be extra constraints of that type.

In the following we describe how Schaefer constructs an instance of \stretchability that is positive if and only if a given program in \textrm{RG-NF} is positive.

\subparagraph{Step 1: The Dependence Graph and Its Embedding} 
Based on the \textrm{GR-NF} program, Schaefer builds the \emph{dependence graph} $G$ which contains a vertex for  every variable $X_1,\ldots,X_n,V_1,\ldots,V_m$ and a special vertex $s$ representing the unit value $1$. In $G$, there is an edge to the computed variable $V_i$ from each variable occurring in its computation as well as from $s$ (the connection from $s$ will be used to ensure that the unit value is represented the same way in all variables); see Fig.~\ref{fig:graph_degree_a}. Moreover, there is an edge from $V_i$ to $V_j$ if $V_i < V_j$ is required. Graph $G$ is used in the reduction to compute positions for the gadgets actually computing $V_1,\ldots,V_m$ as well as fixing how information is routed between gadgets. To this end, Schaefer first computes a drawing of $G$ that fits in a polynomial sized grid. This is an important step in the reduction as it ensures that $G$ admits a easy to compute drawing, i.e., a positive instance will admit a certificate segment arrangement.

\begin{figure}
    \centering
    \includegraphics[page=1]{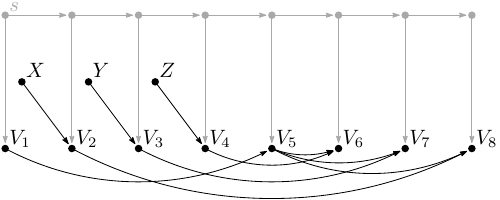}
    \caption{
    An example graph $G$ (example from \cite{S21b}) for a program in Richter-Gebert normal form, where $V_1 = 1$, $V_2=X$, $V_3=Y$, $V_4=Z$, $V_5=-V_1$, $V_6 = -V_5+V_4$, $V_7 = -V_5+V_3$ and $V_8 = -V_5+V_2$, resulting in a degree 5 vertex labeled $V_5$.
    }
    \label{fig:graph_degree_a}
\end{figure}

More precisely, Schaefer first computes a $1$-bend orthogonal drawing of $G$ using the following result by Biedl and Kaufmann:

\begin{theorem}[Biedl,Kaufmann~\cite{DBLP:conf/esa/BiedlK97}]\label{thm:biedl}
    Every connected graph on $n$ vertices and $m$ edges has a $1$-bend orthogonal drawing on an $\frac{n+m}{2} \times \frac{n+m}{2}$ grid. Each vertex is represented by a rectangle of perimeter at most twice its degree.
\end{theorem}

It is worth remarking that orthogonality is only required so to satisfy the constraints for the next step of the drawing procedure. Namely, a result by Brass et al. is applied to guarantee that no three vertices are collinear and that no two vertices use the same $x$- or $y$-coordinate:

\begin{theorem}[Brass, et al.~\cite{DBLP:journals/comgeo/BrassCDEEIKLM07}]\label{thm:brass}
    Suppose $H$ has a planar straight-line drawing in an $O(k)\times O(k)$ grid which has unit edge resolution.
    Then there is a straight-line drawing of $H$ in an $O(k^2)\times O(k^2)$ grid in which no three vertices of $H$ are collinear and no two vertices of $H$ lie on the same grid-line. The new drawing of $H$ can be found in polynomial time, and does not depend on the edges in $H$, just the position of its vertices.
\end{theorem}

For our adjustments in \cref{sec:chromaticnumber}, we point out the following two observations that are immediate from Schaefer's proof.

\begin{observation}\label{obs:subdivision}
    If a vertex in $G$ has high outdegree, its outgoing edges can be replaced by a binary tree along which information can be transmitted (see e.g. vertex $s$ in \cref{fig:graph_degree_a}). The same is true for high indegree.
\end{observation}

\begin{observation}\label{obs:no-collinear}
Any embedding of dependence graph $G$ that fits in a polynomial sized grid such that no two vertices lie on the same grid-line and such that no three vertices are collinear, is a valid output of Step 1.     
\end{observation}

\subparagraph{Step 2: Von Staudt Gadgets and Transmission-Gadgets}
In the next step, the vertices and edges of $G$ are replaced by suitable gadgets. For the vertices, Schaefer uses a well established method to represent the instructions and conditions with so-called \emph{von Staudt gadgets}; see \cite{Richter1995}.
A von Staudt gadget consists of a baseline $\ell$, which contains a number of labeled points, which include three points labeled $\infty$, $0$ and $1$ as well as points representing the values of variables involved in an instruction; \cref{fig:gad_intro}.
The value of a variable is determined by its relative placement and distance to the three points $\infty$, $0$ and $1$.
Additionally a number of pseudo-segments are added, whose intersection pattern guarantees that in a stretched version of the gadget, i.e., one where all pseudo-segments have become straight-line segments, the position of the point representing the variable whose value is computed with the instruction is determined by the positions of the points for the variables which are involved in this computation.
Depending on the specific instruction that is being represented, the exact structure of the arrangement differs. However, there are essentially only two different types of gadgets: one encoding the (negated) addition (Type \ref{type:negated-addition}; see \cref{fig:gad_intro:add}) and (inverted) multiplication (Type \ref{type:invertedMultiplication}; see \ref{fig:gad_intro:mult}). The gadgets for other types of constraints are obtained by slight modifications: Assignment (Types \ref{type:unit} and \ref{type:equality}) simply use line $\ell$ where for the unit assignment (\ref{type:unit}), we additionally have $1=x$ on $\ell$.  Negation is obtained from the (negated) addition gadget by setting $y=0$ and inversion is  obtained from the (inverted) multiplication gadget by setting $y=1$. Finally, the condition $x_1 < x_2$ is checked with a (negated) addition gadget with $y=0$ and $-x=-x_1$ and $x+y=x_2$. 

\begin{figure}
    \centering
    \begin{subfigure}[t]{.4\linewidth}
        \centering
        \includegraphics[page=1]{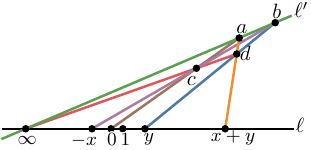}
        \subcaption{}
        \label{fig:gad_intro:add}
    \end{subfigure}
    \quad
    \begin{subfigure}[t]{.4\linewidth}
        \centering
        \includegraphics[page=2]{figures/gadgets-intro.pdf}
        \subcaption{}
        \label{fig:gad_intro:mult}
    \end{subfigure}
    \caption{Von Staudt gadgets: (a) (Negated) Addition and (b) (Inverted) Multiplication.}
    \label{fig:gad_intro}
\end{figure}

A variable appears in multiple gadgets, namely the one of the instruction determining its value, all following instructions, where its value is used to compute a later computed variable and all conditions involving this variable. To guarantee the same value of a variable in all its occurrences, Schaefer uses \emph{transmission gadgets} that are routed along the edges of $G$. The transmission gadgets shown in \cref{fig:transmission-intro} effectively copies the relative position of three points from one straight-line segment $\ell$ to another one $\ell'$ while at the same time negating the encoded information. By using an even number of transmission gadgets, the negation of the transmitted information can be ignored. Moreover, the gadget is stretchable, i.e., the relative position and rotations of $\ell$ and $\ell'$ are independent as long as the topology is maintained.This allows us to route it along the edges of $G$.

\begin{figure}
    \centering
    \includegraphics[page=1]{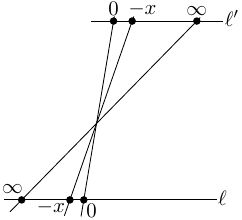}
    \caption{
    A transmission gadget transmitting variable $x$.
    }
    \label{fig:transmission-intro}
\end{figure}

\subparagraph{Step 3: Making the Arrangement Constructible}
So far, more than two pseudo-segments may be required to cross at the same point; see \cref{fig:gad_intro,fig:transmission-intro}. Schaefer discusses how to avoid this circumstance to prove that realizability of a partial abstract order is \ER-complete. Conveniently, we also require this property for our construction in \cref{sec:geothickness}. Schaefer applies here a replacement technique by Las Vergnas~\cite{Vergnas86}. Namely, if a pseudo-segment $u$ crosses with two pairwise crossing pseudo-segments in the same point, it is replaced by two pairwise crossing segment $u_1$ and $u_2$ as shown in \cref{fig:replacement_intro:two}. Moreover, if $u$ contains the crossings points of at least two pairs of such pseudo-segments, it is instead replaced by four segments $u_1$, $u_2$, $u_3$ and $u_4$ as shown in \cref{fig:replacement_intro:four}. 
\begin{figure}
    \centering
    \begin{subfigure}[t]{1\linewidth}
        \centering
        \includegraphics[page=1]{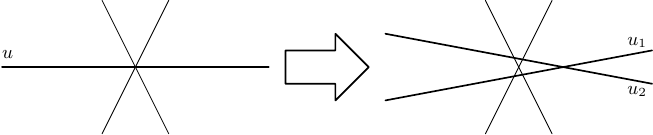}
        \subcaption{}
        \label{fig:replacement_intro:two}
    \end{subfigure}
    \quad
    \begin{subfigure}[t]{1\linewidth}
        \centering
        \includegraphics[page=2]{figures/lasvergnas-intro.pdf}
        \subcaption{}
        \label{fig:replacement_intro:four}
    \end{subfigure}
    \caption{Making a pseudo-segment arrangement uniform.}
    \label{fig:replacement_intro}
\end{figure}
Schaefer discusses that applying this transformation is sufficient to make the pseudo-line arrangement constructible (i.e., no three segments cross at the same point) which concludes the reduction from \textsc{Strict Inequality} (restricted to \textrm{RG-NF} inputs) to \stretchability establishing \ER-hardness of the latter problem.

\section{Geometric Thickness is \texorpdfstring{$\exists \mathbb{R}$}{ER}-Complete}
\label{sec:geothickness}

In this section, we first show \ER-membership.
Then we show \ER-hardness, splitting the proof into construction, completeness, and soundness.

\paragraph{\ER-membership}
To prove that \geo is in \ER, we follow the characterization provided by Erickson, Hoog, and Miltzow~\cite{EvdHM20}. According to this characterization, a problem \( P \) is in \ER if and only if there exists a real verification algorithm \( A \) for \( P \) that runs in polynomial time on the real RAM. Moreover, for every yes-instance \( I \) of \( P \), there should be a polynomial-size witness \( w \) for which \( A(I, w) \) returns ``yes''. For every no-instance \( I \) and any witness \( w \), \( A(I, w) \) should return ``no''.

We now describe a real verification algorithm \( A \) for \geo. Given an instance \( I = (G,t) \), 
the witness \( w \) consists of the coordinates of the vertices of $G$ and a \( t \)-coloring of the edges $\chi:E \rightarrow {[t]}$.

 \( A \) then verifies that there are no monochromatic crossings in the induced drawing by examining every pair of edges \( e \) and \( e' \) that cross. 
 Furthermore, $A$ checks that all edges with higher multiplicity receive distinct colors.
 This verification can be done on a real RAM in \( O(|E|^2) \) time.

Hence, \( A \) runs in polynomial time on the real RAM, confirming that \geo is in \ER.

\begin{figure}[tbp]
    \centering
    
    \includegraphics[page = 6]{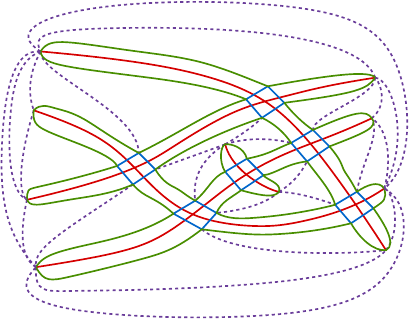}
    \hfil
    \includegraphics[page = 7]{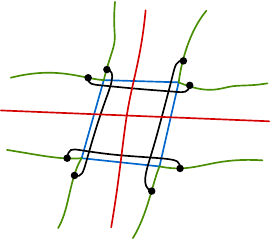}

    \caption{Left: The pseudo-segment arrangement is drawn in red. 
    The tunnel paths are drawn in green.
    The crossing boxes are drawn in blue and the blockers are not drawn.
    The connectors are drawn in purple (dotted).
    Right: Depiction for adding the blockers. }
    \label{fig:construction}
\end{figure}

\paragraph{Construction}
We reduce \segmentStretchability to \geo. 
For an illustration of this reduction, 
we refer to \Cref{fig:construction}.
Let \A be a given 
pseudo-segment arrangement with $n$ pseudo-segments
and let $t \geq \chromaticnumber$ be the geometric thickness we are aiming for.
We construct a graph $G$ as follows.
\begin{itemize}[leftmargin=*]
    \item For each pseudo-segment $a$ in \A, we add one \emph{long edge} (just a single edge) $(a_1,a_2)$ with multiplicity 1. More precisely, $a_1$ and $a_2$ are placed at the ends of $a$.
    \item For each crossing, we add a \emph{crossing box} ($4$-cycle)  with multiplicity $t-1$. More precisely, if pseudo-segments $a$ and $a'$ cross, we create the $4$-cycle $(c_{a_1,a'_1},c_{a_1,a'_2},c_{a_2,a'_2},c_{a_2,a'_1})$.
    \item   We connect the endpoints of each long edge with the corresponding crossing boxes by \emph{tunnel boundaries}
     (paths of length $n$). More precisely, consider a pseudo-segment $a$ and let $a'$ be the pseudo-segment $a$ closest to its end $a_i$. Then, we connect $a_i$ with $c_{a_i,a'_1}$ and $c_{a_i,a'_2}$ via a tunnel boundary.
     Moreover, we also connect consecutive crossing boxes via tunnel boundaries. More precisely, let $a$ be a pseudo-segment and $a'$ and $a''$ be two pseudo-segments crossing $a$ such that the crossings with 
 $a'$ and $a''$ are consecutive when traversing $a$ from $a_1$ to $a_2$. Moreover, assume that the endpoints of $a'$ and $a''$ are labeled such that $a'_1$ and $a''_1$ occur on the same side of $a$. Then, we connect $c_{a_2,a'_j}$ and $c_{a_1,a''_j}$ with a tunnel boundary for $j \in \{1,2\}$.
     All tunnel boundaries have multiplicity $t$.
    \item We add four \emph{blockers} (single edges with multiplicity 1) to each crossing box, s.t. they connect corresponding tunnel boundaries around each crossing box and two blockers cross each pair of opposite crossing box edges. More precisely, consider a crossing box $(c_{a_1,a'_1},c_{a_1,a'_2},c_{a_2,a'_2},c_{a_2,a'_1})$ corresponding to a crossing between $a$ and $a'$. Moreover, let $b_{a_i,a'_j}$ and $b'_{a_i,a'_j}$ denote the vertex incident to $c_{a_i,a'_j}$ on a tunnel boundary connected to another crossing box or an endpoint of $a$ or $a'$, respectively. Then, we add the four blockers $(b_{a_1,a'_1},b_{a_2,a'_1})$, $(b_{a_1,a'_2},b_{a_2,a'_2})$, $(b_{a_1,a'_1},b_{a_1,a'_2})$ and $(b_{a_2,a'_1},b_{a_2,a'_2})$.
    \item We add \emph{connectors} (paths of length~$n$) of multiplicity $t$ between the endpoints of the long edges and the crossing boxes. First, we add a cycle of them between the endpoints of the long edges on the outer face. Then, we ``triangulate'' the faces that are bounded by tunnel boundaries and connectors by adding connectors such that each such face is incident to three vertices which are endpoints of either long edges or edges belonging to crossing-boxes. In other words, each such face becomes a triangle after contracting all edges belonging to the tunnel boundaries and connectors.
    Furthermore, we add the paths in such a way that 
    no triangle contains both endpoints of a long edge.
\end{itemize}

This finishes the construction. 
Clearly, it can be done in polynomial time.
Note that the illustrations of the graph $G$ in \Cref{fig:construction} are meant 
to help to understand the definition of $G$.
However, we do not know (yet) that the graph $G$ 
needs to be embedded in the way depicted in the figures.

\paragraph{Completeness}
\label{sub:completeness}
Let \S be an arrangement of (straight-line segments) such that \S is isomorphic to \A. 
We need to show that there is a drawing of $G$ such that we can color it 
with $t$ colors avoiding monochromatic crossings.

We start by adding all the edges described in the construction. 
Note that all added parts are very long, thus, it is easy to see that they offer enough flexibility to realize each edge as a straight-line segment. 
Next we need to describe a coloring of all edges.
The tunnel edges have multiplicity~$t$, but they cross no other
edges so this is fine.
We know that we can color the long edges with $t$ colors due to
\Cref{cor:coloring}.
Each crossing box edge has multiplicity $t-1$. We use the $t-1$ colors
different from the color of the long edge crossing it.
Each crossing box has four edges. The two opposite edges receive the same set of colors.
The blockers receive the colors of the corresponding long edge that crosses the same crossing box edges.
It can be checked that no two edges that cross or overlap have the same color.

\paragraph{Soundness}
We now argue that if $G$ has geometric thickness $t$ then \A is stretchable.
Let $\Gamma$ be a straight-line drawing of $G$ 
with a $t$-coloring of $E$.

\paragraph*{Frame}
Let the \emph{frame} $H$ of $G$ be the subgraph of $G$ that consists only of the crossing boxes, tunnel boundaries, and connectors. Since all edges in $H$ have multiplicity at least $t-1$, no two edges may cross, so $H$ has to be drawn planar.
As a first step, we will argue that the frame of $G$ has a unique combinatorial embedding. It is well-known that a planar graph has a unique combinatorial embedding if and only if it is the subdivision of a planar 3-connected graph~\cite{Nishizeki1988PlanarGT}.

\begin{lemma}
\label{lem:frame}
The frame $H$ of $G$ has a unique combinatorial embedding.
\end{lemma}
\begin{proof}
First, consider the \emph{contracted frame} $H^*$ obtained as follows; see \cref{fig:frame-contracted}. First, contract each tunnel boundary and each connector to a single edge. Then, contract each crossing box to a single vertex and remove multi-edges that appear by the contraction. By construction of the connectors, the resulting graph $H^*$ is a planar triangulated graph where each vertex $v^*$ corresponds to either the endpoint of a long edge or to a crossing box in $H$, and each edge $e^*$ corresponds to either a connector or to two boundary paths in $H$. Since $H^*$ is a planar triangulated graph, it has a unique embedding (up to mirroring).

\begin{figure}[t]
    \centering
    \subcaptionbox{The contracted frame $H^*$\label{fig:frame-contracted}}{\includegraphics[page=6]{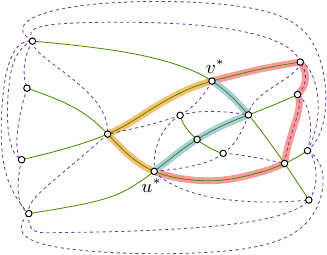}}
    \hfil
    \subcaptionbox{Three vertex-disjoint paths in $H$\label{fig:frame-paths}}{\includegraphics[page=8]{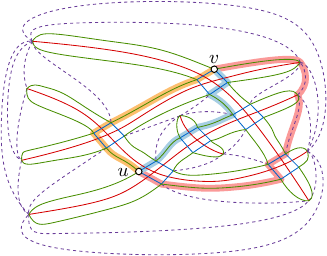}}
    \label{fig:frame}
    \caption{Illustration for the proof of \cref{lem:frame}.}
\end{figure}

We now argue that $H$ also has a unique embedding by proving that it is the subdivision of a planar 3-connected graph. Obviously, $H$ has no vertices of degree~1. Hence, we only have to prove that there are three vertex-disjoint simple paths between any two vertices of degree at least~3. Let $u,v$ be two such vertices of $H$.  

\looseness=-1
First, assume that $u$ and $v$ belong to the same crossing box. If $u$ and $v$ are adjacent, then they share two common faces. One of the two faces is the interior of the crossing box, the other face is bounded by two tunnel boundaries and the endpoint of a long edge or an edge of a different crossing box. We obtain the three paths from the edge $(u,v)$ and from following the boundaries of the common faces. If $u$ and $v$ are not adjacent, we find two paths following the edges of the crossing box. For the third path, let $w^*$ be the corresponding vertex in $H^*$, let $e_u$ and $e_v$ be tunnel paths incident to $u$ and $v$, respectively, and let $e^*_u$ and $e^*_v$ be the edges of $H^*$ that correspond to $e_u$ and $e_v$, respectively. Since $H^*$ is 3-connected, there exists a cycle $C^*=\langle w^*, e^*_0, w_1^*,e_1^*,\ldots, e_{k-1}^*,w_k^*,e_k^*,w^*\rangle$ in $H^*$ through $e_u^*=e_0^*$ and $e_v^*=e_k^*$. We find a path $P$ from $u$ to $v$ in $H$ by replacing each edge $e_i^*$ of $C^*$ by the corresponding tunnel path or connector in $H$, and each vertex $w_k^*$ by the corresponding vertex or a path through the corresponding crossing box.

\looseness=-1
Now, assume that $u$ and $v$ do not belong to the same crossing box. Let $u^*$ and $v^*$ be the vertices of $H^*$ corresponding to $u$ and $v$, respectively. Let $P_1^*,P_2^*,P_3^*$ be three vertex-disjoint paths between $u^*$ and $v^*$ in $H^*$; see \cref{fig:frame-contracted}. Similar to the previous case, we aim to find three paths $P_1,P_2,P_3$ by replacing the edges of $P_1^*,P_2^*,P_3^*$ by corresponding tunnel paths or connectors and interior vertices by paths through corresponding crossing boxes, if necessary. Since the paths are vertex-disjoint, we do not visit any crossing boxes more than once, except those of $u$ and $v$. In fact, if $u$ is a vertex of a crossing box, it might happen that the tunnel paths that correspond to the first edges of $P_1,P_2,P_3$ do not start in $u$, but in different vertices of the crossing box of $u$.

To get rid of this problem, we now prove a slightly stronger version of Menger's theorem: 
(We believe that this statement might have been proven before. We include a proof here for completeness.)
\begin{claim}
Let $G$ be a 3-connected graph with vertices $u$ and $v$ and edges $e_u=(u,u')$ and $e_v=(v,v')$. Then there exist three interior-vertex-disjoint paths $P_1,P_2,P_3$ from $u$ to $v$ such that $e_u,e_v\in P_1\cup P_2\cup P_3$. 
\end{claim} First, find three interior-vertex-disjoint paths $P_1',P_2',P_3'$ from $u$ to $v$. If these paths contain $e_u$ and $e_v$, we are done. Otherwise, assume that they do not contain $e_u$. By 3-connectivity, there are three vertex-disjoint paths from $u'$ to $v$, so at least one of them does not visit $u$. Let $Q$ be this path. If $Q$ is interior-vertex-disjoint from two of $P'_1,P'_2,P'_3$, say $P'_2$ and $P'_3$, then $P_1=\langle u,e_u,Q\rangle,P'_2,P'_3$ are three interior-vertex-disjoint paths from $u$ to $v$ with $e_u\in P_1$. Otherwise, follow $Q$ until it reaches an interior vertex  of $P'_1,P'_2,P'_3$ for the first time, say vertex $w$ on $P'_1$. Create a path $P_1$ from $u$ to $v$ by following $e_u$, then $Q$ until reaching $w$, then $P'_1$ until reaching $v$. Then $P_1,P'_2,P'_3$ are three interior-vertex-disjoint paths from $u$ to $v$ with $e_u\in P_1$. We can force $e_v$ to be part of one of the paths analogously.

Hence, we can assume that at least one of $P_1,P_2,P_3$ starts with a tunnel path at $u$ and at least one of $P_1,P_2,P_3$ ends with a tunnel path at $v$; see \cref{fig:frame-paths}. For the other two paths at $u$, we can reach the endpoint of the first tunnel path by following the crossing box, if necessary. An analogous argument can be applied for the last edges to reach $v$.
\end{proof}

As a next step, we will define tunnels precisely.

\paragraph*{Tunnels}
In \Cref{fig:tunnel-segment}, we illustrate the regions in the plane that
we refer to as \emph{middle tunnel segments}, \emph{end tunnel segments}, and \emph{crossing boxes}. A \emph{middle tunnel segment} is the region delimited by two crossing box edges and two tunnel boundaries which start and end at the corresponding crossing box edges. A \emph{end tunnel segment} on the other hand is the region delimited by the two tunnel boundaries starting at an endpoint of a long edge and the crossing box edge where these boundaries end.
We believe that those notions are intuitive to understand from the figure and thus we avoid a more formal definition. (End or middle) tunnel segments incident to the same crossing box are called \emph{consecutive}.

\begin{figure}[tb]
    \centering
    \includegraphics[width=0.9\linewidth]{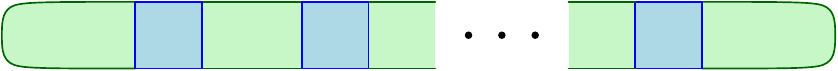}
    \caption{Each tunnel consists of two end tunnel segments and various middle tunnel segments, separated by crossing boxes. 
    }
    \label{fig:tunnel-segment}
\end{figure}

We say that a tunnel segment has \emph{color} $c$, if it contains a long edge with color $c$.
We will show that each tunnel segment has exactly one color.

\begin{myclaim}
Each end tunnel segment contains at least one long edge. 
\end{myclaim}
Consider any long edge $\ell$ and the frame $H'$ without the blockers, and any embedding of $H'\cup\{\ell\}$. Since the frame has a unique combinatorial embedding and a long edge cannot cross a tunnel boundary or connector (as they have multiplicity~$t$), this embedding must be plane and coincide with the unique combinatorial embedding of $H$ on $H'$. The only way to add $\ell$ into the embedding of $H'$ is through its end tunnel segment.

\begin{myclaim}
    Each tunnel segment has at most one color $c$
    and its bounding crossing box edges have all the remaining colors.
\end{myclaim}
Note that each tunnel segment can only be entered or left by a long edge using the multiedges 
of the crossing box. 
As those have multiplicity $t-1$, it follows that the edges of the crossing box are colored with
all colors except the color of the long edge.
Now, we see that each tunnel segment is completely surrounded by edges of all but one colors, which implies the claim.

\begin{figure}[t]
    \centering
    \includegraphics[page = 3]{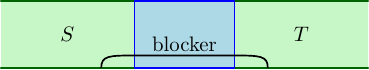}
    \caption{Consecutive tunnel segments must have the same color.}
    \label{fig:tunnel-colors}
\end{figure}

\begin{myclaim}
    Consecutive tunnel segments have the same color.
\end{myclaim}
See \Cref{fig:tunnel-colors}, for an illustration.
Denote by $S,T$ the two consecutive tunnel segments and by $C = [t]$ the set of all colors. 
Say tunnel $S$ has color $c$.
We show that tunnel $T$ has the same color.
Due to the unique embedding there is a blocker
that intersects both tunnels. 
As the blocker is in a tunnel with color $c$,
the bounding crossing box edges have colors $C\setminus c$.
Thus, the blocker must have color $c$ as well.
Hence, also $T$ must have color $c$ as all bounding crossing box edges must have 
colors $C\setminus c$ (unless they are uncrossed).
This proves the claim and also immediately the following:

\begin{myclaim}
    All tunnel segments of a tunnel have the same color.
\end{myclaim}

We are now ready to show our central claim:

\begin{myclaim}\label{claim:insideTunnel}
    Each long edge
stays in its respective tunnel.
\end{myclaim}
First, we note that all the blockers of a tunnel have the same color as the long edge.
Furthermore, due to the combinatorial embedding, the tunnel paths together with the blockers
form a cycle surrounding the long edge. 
It remains to show that the blockers cannot leave their respective tunnels,
to show that the long edges cannot leave the tunnel.
To this end, notice that a blocker cannot cross a crossing box edge twice,
because they are both line segments and two line segments can cross at most once.

\begin{myclaim}
    \label{claim:combiEquiv}
    The arrangement \S formed by the long edges is combinatorially equivalent 
    to the arrangement \A of pseudo-segments.
\end{myclaim}

This follows from the tunnels crossing 
combinatorially as in arrangement \A and the long edges
staying within their respective tunnels.

\Cref{claim:combiEquiv} concludes the proof of \Cref{thm:para-ER}. 
Notice that in our construction, for every crossing, we add $2t$ more crossings ($2(t-1)$ from the parallel crossing box edges and two from the blockers). 
Since the pseudo-segment graph from Schaefer's construction has maximum degree 72~\cite{S10}, the graph in our reduction is $\kplanarnumber$-planar, which proves \Cref{cor:kpla-ER}.

\section{\boldmath Sunflower \sge}
\label{sec:SGE}
This section is devoted to proving \Cref{thm:SGE} by slightly modifying our construction in the proof of \Cref{thm:para-ER}. 
We reduce  \segmentStretchability to \sge with the additional restriction that the input graphs of the \sge instance form an empty sunflower. To this end, let $\mathcal{A}$ denote a pseudo-segment arrangement with $n$ pseudo-segments. Further, let $c$ be the chromatic number of the pseudo-segment intersection graph induced by $\mathcal{A}$ and let $\chi:\mathcal{A} \rightarrow [c]$ denote a corresponding $c$-coloring.

We construct an instance of \sge consisting of $k:=c+1$ simple graphs $H,G_1,\ldots,G_c$ on a shared vertex set as follows. The graph $H$ contains exactly the edges belonging to the frame graph as defined above, 
i.e., the crossing boxes, the tunnel boundaries and the connectors. Moreover, for $i\in [c]$, the graph $G_i$ contains the long edges corresponding to each pseudo-segment $S \in \mathcal{A}$ for which $\chi(S)=i$ and a $1$-subdivision of the edges of $H$ aside from the crossing box edges bounding tunnel segments corresponding to pseudo-segments $S \in \mathcal{A}$ for which $\chi(S)=i$. More precisely, in the $1$-subdivision of the subgraph of $H$ belonging to $G_i$, each edge $e=(u,v)$ of $H$ is replaced by a path $(u,x_i(e),v)$ of length $2$ where $x_i(e)$ does not belong to $H$, i.e., $x_i(e)$ is an isolated vertex in all graphs except for $G_i$. Since $\chi$ is a proper $c$-coloring, the graphs $G_1,\ldots,G_c$ do not share any long edges, while their $1$-subdivisions of $H$ are edge-disjoint by construction. As also $H$ is edge-disjoint from any of the $1$-subdivisions, we observe that $H,G_1,\ldots,G_c$ form an empty sunflower. 
Note that the construction here does not require any blockers.

It remains to discuss that $\mathcal{A}$ is stretchable if and only if $H,G_1,\ldots,G_c$ admit a simultaneous geometric embedding. First, completeness can be easily shown following the argumentation in the corresponding paragraph in \cref{sec:geothickness}. In particular, we need to discuss how to place the subdivision vertices of edges of $H$. Namely, for an edge $e=(u,v)$ of $H$, we can place all subdivision vertices $x_i(e)$ arbitrarily close to the straight-line segment representing $(u,v)$. Completeness now immediately follows by observing that $G_i$ contains no subdivisions of crossing-box edges that bound tunnel segments corresponding to a pseudo-segment $S \in \mathcal{A}$ for which $\chi(S)=i$, i.e., the tunnel of a segment $S\in \mathcal{A}$ with $\chi(S)=i$ is a single face in $G_i$ minus the long edge representing $S$.

\begin{figure}[t]
    \centering
    \includegraphics[width=\linewidth,page=4]{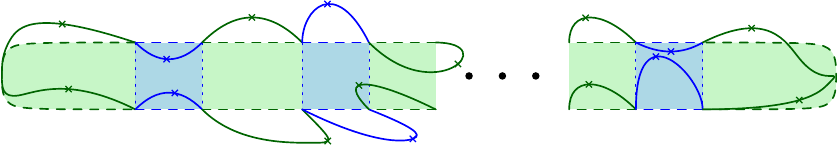}
    \caption{In our construction for proving \Cref{thm:SGE}, subdivided tunnels do not contain subdivisions of crossing box edges shared by consecutive segments. While subdivided tunnels may cover not all of the tunnel and also additional parts of the plane, each curve connecting between both end segments traverses all segments in order.}
    \label{fig:tunnel-segment:sge}
\end{figure}

Finally, we show {soundness}. Let $\varphi$ be a simultaneous geometric embedding of $H,G_1,\ldots,G_c$. For a pseudo-segment $S \in \mathcal{A}$, we define the \emph{subdivided tunnel} as the $1$-subdivision of the \emph{outer cycle} of its tunnel belonging to $G_{\chi(S)}$; see \Cref{fig:tunnel-segment:sge}. Note that $G_{\chi(S)}$ does not contain a subdivision of the crossing box edges shared by consecutive segments of a subdivided tunnel; see dotted edges in \Cref{fig:tunnel-segment:sge}. We now  prove the equivalent of \Cref{claim:insideTunnel}:

\begin{myclaim}
    Each long edge traverses all segments and crossing boxes of its respective tunnel in order. 
\end{myclaim}

First note that by \Cref{lem:frame}, we know that $H$ has a unique combinatorial embedding. Now consider a pseudo-segment $S$ with $\chi(S)=i$. Since $G_i$ contains subdivisions of all tunnel boundaries and all connectors, the corresponding long edge $\ell$ must be drawn completely inside tunnels. By construction, the endpoints of $\ell$ are contained only on the boundary of the subdivided tunnel corresponding to $S$, i.e., it must start and end inside its respective subdivided tunnel.   Moreover, by construction, the entire subdivided tunnel belongs to $G_i$. Thus, $\ell$ cannot enter other tunnels in between. While the subdivided tunnel may be covering a superset of the tunnel, $\ell$ must still traverse all of its segments as the crossing box edges between consecutive segments still separate subdivided segments. 

The fact that each long edge still traverses each crossing box in order implies immediately  \Cref{claim:combiEquiv} in our reduction to \sge and the theorem follows.

\newcommand*{\myrulefill}[3][]{%
\noindent
  \makebox[#2]{#1%
    \leaders\hrule height \dimexpr.5ex+.2pt\relax depth \dimexpr -.5ex+.2pt\relax \hfill
    \enskip{#3}\enskip
    \leaders\hrule height \dimexpr.5ex+.2pt\relax depth \dimexpr -.5ex+.2pt\relax \hfill\kern0pt}
}
\newcommand{\wavefront}{\myrulefill{\linewidth}{\textbf{WF}}\vspace{2cm}}

\section{Chromatic Number of Pseudo-Segment Arrangement}
\label{sec:chromaticnumber}

\begin{figure}
    \centering
    \begin{subfigure}[t]{.3\linewidth}
        \centering
        \includegraphics[page=2]{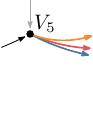}
        \subcaption{}
        \label{fig:graph_degree_b}
    \end{subfigure}
    \begin{subfigure}[t]{.3\linewidth}
        \centering
        \includegraphics[page=3]{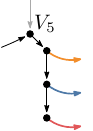}
        \subcaption{}
        \label{fig:graph_degree_c}
    \end{subfigure}
    \caption{ 
    Example of the adaption of the graph shown in Figure~\ref{fig:graph_degree_a} in which the vertex $V_5$ has degree 5.
    The three outgoing edges of $V_5$ (b) are instead replaced by a binary tree (c), s.t., no vertex (in the entire graph) has degree larger than 4.}
    \label{fig:graph_degree}
\end{figure}

In this section, we will show that the pseudo-segment arrangement from Schaefer's construction~\cite{S21b} -- the construction is also explained in Section~\ref{sec:reduction_description} -- has chromatic number $\le 30$.
This is done via small adaptions and a detailed analysis of the constructions in Schaefer's proof.
Recall that the reduction creates the \emph{dependence graph} $G$ based on a Richter-Gebert normal form program.
Schaefer now creates a 1-bend orthogonal drawing of $G$.
Here, we modify  $G$ to obtain a new graph $G'$, which is similar to $G$, but vertices of high degree are replaced by a binary tree, which guarantees that every vertex in $G'$ has degree at most four; see \cref{fig:graph_degree}.
We replace $G$ with $G'$ and continue following the construction.
A 1-bend orthogonal drawing $\Gamma$ of $G'$ on a grid of size $O(n^2) \times O(n^2)$ in which all vertices are boxes of small perimeter is created using Theorem~\ref{thm:biedl}.
Before this drawing is further refined, we make our second adaption.
We change how edges are routed in a small local area around every vertex, s.t., exactly one edge connects to the box from the top, while all other edges connect from the bottom (a choice which can be made independently for every vertex; see \cref{fig:box_redrawing_2}) and the drawing still fits on a grid of size $O(n^2) \times O(n^2)$.
Schaefer now obtains a straight line graph drawing with unit edge resolution based on the orthogonal drawing.
A straight line drawing with the same necessary properties can be obtained based on our adapted drawing.
Based on this drawing, it is possible to create a second straight-line drawing on a grid of size $O(n^4) \times O(n^4)$ and with an additional factor of $n$ to the grid size it can be ensured that such a drawing has unit edge resolution (Theorem~\ref{thm:brass}).
Finally the drawing is replaced with a set of pseudo-segments in the form of von Staudt gadgets at positions of specific vertices and transmission gadgets in place of the edges.
We explain how these pseudo-segments can be colored in a way that no two crossing segments have the same color using only 10 colors.
To this end we will utilize the properties guaranteed by our previous adaptions.
Finally Schaefer uses the dual construction by Las Vergnas~\cite{Vergnas86} to ensure that in a stretched version of the pseudo-segment arrangement, no three crossing points are collinear.
This method increases the chromatic number of the arrangement trivially by at most a factor of four.
We show that in in fact increases it by at most a factor of three resulting in a final chromatic number of \chromaticnumber for the pseudo-segment arrangement.

\paragraph{Guaranteeing low degree in $G$}

\begin{figure}
    \centering
    \begin{subfigure}[t]{.18\linewidth}
        \centering
        \includegraphics[page=6]{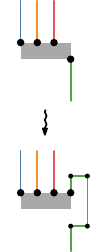}
        \subcaption{}
    \end{subfigure}
    \hfill
    \begin{subfigure}[t]{.18\linewidth}
        \centering
        \includegraphics[page=7]{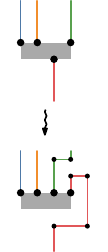}
        \subcaption{}
    \end{subfigure}
    \hfill
    \begin{subfigure}[t]{.18\linewidth}
        \centering
        \includegraphics[page=8]{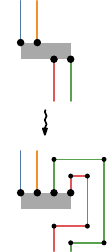}
        \subcaption{}
    \end{subfigure}
    \hfill
    \begin{subfigure}[t]{.18\linewidth}
        \centering
        \includegraphics[page=9]{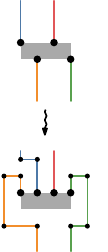}
        \subcaption{}
    \end{subfigure}
    \hfill
    \begin{subfigure}[t]{.18\linewidth}
        \centering
        \includegraphics[page=10]{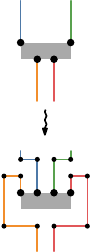}
        \subcaption{}
    \end{subfigure}
    \caption{Any vertex rectangle in an orthogonal drawing can, regardless of the way the edges are connected, be replaced by one, where all edges connect from the top (a-e).
    Every edge is subdivided either 0, 2 or 4 times.
    }
    \label{fig:box_redrawing}
\end{figure}

Since every clause in a program in Richter-Gebert normal form computes the value of a computed variable based on at most two previous values, by construction every vertex in $G$ has at most 3 incoming edges (one incoming edge from the binary tree representing the rootnode $s$ as well as two for the input values).
However since an underlying or a computed variable can possibly be used in many computations of computed variables, with a higher index in $\pi$, some of the vertices in Schaefers construction might have a degree larger than 4.
In order to guarantee that the number of intersection on the pseudo-segments of von Staudt gadgets is low, it needs to be guaranteed that the degree of every vertex is at most four.

\begin{lemma}\label{lem:degree_replacement}
    Any drawing of the graph $G$ can be transformed into a drawing of $G$, in which no vertex has degree larger than 4.
\end{lemma}
\begin{proof}
    This can be achieved by using the same structure that Schaefer uses for $s$, i.e., by replacing the set of outgoing edges of a vertex with a binary tree (see also \cref{obs:subdivision}).
    By construction only the vertices representing computed variables of the program in Richter-Gebert normal form can have a degree larger than 5.
    Let $v$ be a vertex with degree 5 or larger.
    Then $v$ has three incoming edges and a set of outgoing edges (one for each later use of a variable).
    We replace the set of outgoing edges with a binary tree, s.t., $v$ has exactly one outgoing edge, which points to the root of the binary tree and for every vertex $u$ of a later computation or condition using the computed variable represented by $v$ in $G$ there is a unique leaf node of the binary tree connected to $u$.
    This is illustrated in Figure~\ref{fig:graph_degree}.
\end{proof}

\paragraph{Guaranteeing the connection direction for incoming edges}

\begin{figure}[t]
    \centering
    \begin{subfigure}[t]{.2\linewidth}
        \centering
        \includegraphics[page=2]{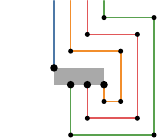}
        \subcaption{}
    \end{subfigure}
    \begin{subfigure}[t]{.2\linewidth}
        \centering
        \includegraphics[page=3]{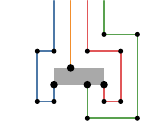}
        \subcaption{}
    \end{subfigure}
    \begin{subfigure}[t]{.2\linewidth}
        \centering
        \includegraphics[page=4]{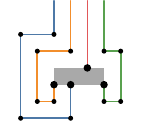}
        \subcaption{}
    \end{subfigure}
    \begin{subfigure}[t]{.2\linewidth}
        \centering
        \includegraphics[page=5]{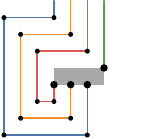}
        \subcaption{}
    \end{subfigure}
    \caption{Any vertex rectangle in an orthogonal drawing to which all edges are connected from the top can be transformed into a drawing, where exactly one (arbitrary) edge connects on the top of the rectangle and all others connect at the bottom.
    Every edge is subdivided either 0 or 4 times.
    }
    \label{fig:box_redrawing_2}
\end{figure}

After application of Theorem~\ref{thm:biedl}, we obtain a drawing in which every vertex is represented by a rectangle with a perimeter of at most twice its degree.
We change this drawing so that one edge can be chosen to connect to the rectangle from the top, while all other connections are made from the bottom. This will come at the cost of increasing the number of bends per edge (by \cref{obs:no-collinear}, this is not important for the remainder of the discussion). 

\begin{lemma}\label{lem:connection_direction}
    Any $1$-bend orthogonal drawing of a connected graph $H$ on a $O(n^2)\times O(n^2)$ grid, where each vertex is represented by a rectangle of perimeter of at most twice the degree of the vertex can be transformed into a $9$-bend orthogonal drawing of $H$ on a $O(n^2)\times O(n^2)$ grid, s.t., every rectangle representing a vertex is connected to exactly one edge on its top side and all other edges on its bottom side.
    Moreover, we can freely decide which edge connects from the top independently for every vertex and every edge is subdivided an odd number of times.
\end{lemma}
\begin{proof}
Using Lemma~\ref{lem:degree_replacement} we ensure that the maximum degree in $H$ is four.
Let $v$ be a vertex of $H$.
We label the edges connected to $v$ in clockwise order as $a, b, c$ and $d$.
It is possible to move all edges, s.t., they connect to the top of the box representing $v$ (Figures~\ref{fig:box_redrawing}a-e) by adding 0, 2 or 4 bends to every edge and adding at most a constant number of additional grid lines per vertex.
Then we can transform this state into one of the four configurations shown in Figures~\ref{fig:box_redrawing_2}a-d (depending on which edge should be connected on the top), which adds either 0 or 4 bends per edge and can be done by adding again only a constant number of grid lines per vertex. 
Note that the circular order of connection points around the rectangle is fixed and never changed during these operations.
\end{proof}

\paragraph{Lowering the chromatic number of the von Staudt and transmission gadgets}

We now investigate the von Staudt gadgets for every computation and condition individually to show how they can be colored locally.
Note that we only consider the construction of the gadgets and how the segments can be colored in order to avoid mono-colored crossings.
For a detailed explanation of the correctness of these gadgets, we refer to Schaefer~\cite{S21b} and Richter-Gebert~\cite{Richter1995}.
When comparing the illustrations of the gadgets below to the illustrations found by Schaefer and Richter-Gebert, it is important to keep in mind that we are also illustrating the connected transmission gadgets.
This is important in particular, because all pseudo-segments of the von Staudt gadgets are colored with pairwise different colors and we only save colors (so to speak) by reusing these colors in the connected transmission gadgets.

\begin{figure}
    \centering
    \begin{subfigure}[t]{.4\linewidth}
        \centering
        \includegraphics[page=17]{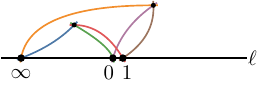}
        \subcaption{$V_i = 1$}
        \label{fig:gad_assignment_1}
    \end{subfigure}
    \quad
    \begin{subfigure}[t]{.4\linewidth}
        \centering
        \includegraphics[page=16]{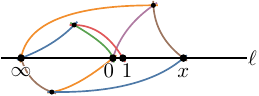}
        \subcaption{$V_i = x$}
        \label{fig:gad_assignment_x}
    \end{subfigure}
    \caption{Coloring of the gadgets to encode the assignment of $1$ (a) or $x$ (b) to a computed variable $V_i$.}
    \label{fig:gad_assignment}
\end{figure}

\paragraph{Assignment}
This is a computation in which a computed variable is set to the value of 1 or of an underlying variable $X_j$, i.e., $V_i = 1$ or $V_i = X_j$.
As Schaefer points out, this can be done with three or four points on a line, respectively.
We can (but do not need to) reuse any colors when coloring the pseudolines of these gadgets, which are illustrated in Figures~\ref{fig:gad_assignment_1} and~\ref{fig:gad_assignment_x}.
Note that in the first gadget, since all pseudolines are colored in pairwise different colors, we do not need to consider if the transmission gadgets are connected from above or below $\ell$.
The second gadget uses the same color for pseudolines of two different transmission gadgets.
Since such lines might cross, if the transmission gadgets connect from the same side of $\ell$, i.e., both from above or both from below.
Therefore we have to guarantee the two specific transmission gadgets are connecting from two different sides.
This guarantee can be made by using the rerouting around vertices as described above.

\paragraph{(Negated) Addition}
The negated addition is a computation, which assigns the value of the sum of two variables (underlying or computed) to a new computed variable.
To compute the sum $x+y$, the gadget requires the input of the values $-x$ and $y$.
The negated value of $-x$ can be obtained with the negation gadget shown in the next paragraph.
As shown in Figure~\ref{fig:gad_addition}, the pseudolines of the von Staudt gadget (bold drawn curves) are colored in 7 pairwise different colors and three of the four connected transmission gadgets (thin curves) entirely reuse these colors, requiring only three additional colors for the fourth transmission gadget (dashed curves) and resulting in a total of 10 required colors.
We highlight the correctness of this coloring by showing the individual ``pages'', i.e., pseudolines with the same color in Figure~\ref{fig:gad_addition_pages}.
It is also important to note that the coloring is still valid as shown if the point labeled $y$ is identified with either the point $0$ or the point $1$, which correspond to an alternative variant of the assignment gadget (as described by Schaefer) and the sum $x+1$, respectively.

\begin{figure}
    \centering
        \includegraphics[page=1]{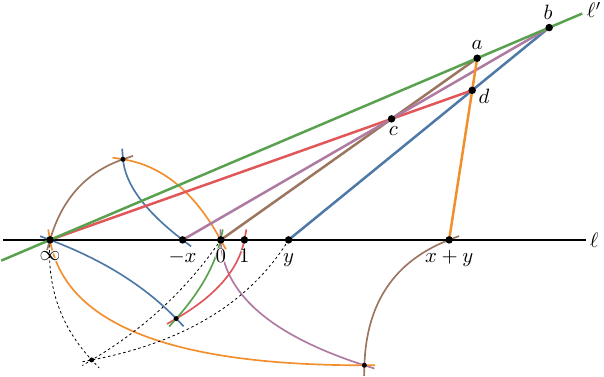}
    \caption{Coloring of the (inverted) addition gadget with four connected inversion gadgets.
    The dashed lines indicate that this inversion gadget is not reusing colors from the gadget, but instead introducing three entirely new colors.}
    \label{fig:gad_addition}
\end{figure}

\begin{figure}
    \centering
    \begin{subfigure}[t]{.4\linewidth}
        \centering
        \includegraphics[page=2]{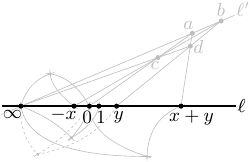}
        \subcaption{}
    \end{subfigure}
    \quad
    \begin{subfigure}[t]{.4\linewidth}
        \centering
        \includegraphics[page=3]{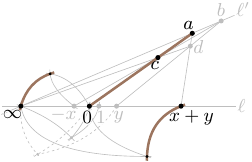}
        \subcaption{}
    \end{subfigure}
    \begin{subfigure}[t]{.23\linewidth}
        \centering
        \includegraphics[page=4]{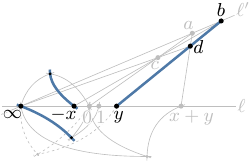}
        \subcaption{}
    \end{subfigure}
    \hfill
    \begin{subfigure}[t]{.23\linewidth}
        \centering
        \includegraphics[page=5]{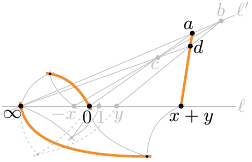}
        \subcaption{}
    \end{subfigure}
    \hfill
    \begin{subfigure}[t]{.23\linewidth}
        \centering
        \includegraphics[page=6]{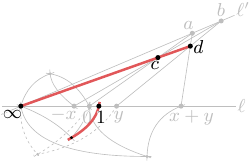}
        \subcaption{}
    \end{subfigure}
    \begin{subfigure}[t]{.4\linewidth}
        \centering
        \includegraphics[page=7]{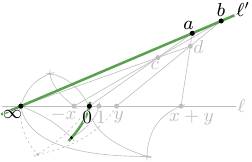}
        \subcaption{}
    \end{subfigure}
    \quad
    \begin{subfigure}[t]{.4\linewidth}
        \centering
        \includegraphics[page=8]{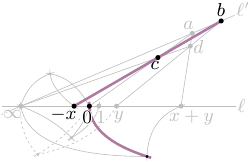}
        \subcaption{}
    \end{subfigure}
    \caption{Illustration of the different ``pages'' containing segments colored in the same color.
    Note that no page contains a mono-colored crossing.
    }
    \label{fig:gad_addition_pages}
\end{figure}

\paragraph{Negation}
Since one of the values of the negated addition gadget is required to be a negative value, the von Staudt gadget for obtaining the negative value of a given value $x$ can be seen as a special case of the negated addition gadget for the computation $x+0$.
Specifically the second value is identified with the point representing the value 0.
Nevertheless a similar coloring as for the addition gadget is possible.
Since a negation only requires the scale and one input value as well as the output value to be transmitted to or from the gadget, there is no fourth transmission gadget needed and a similar coloring as for the negated addition gadget results in 7 required colors (see Figure~\ref{fig:gad_negation}). 

\paragraph{Multiplication}
Similar to the addition gadget, the multiplication gadget (representing the computation assigning the value of a product $x\cdot y$ to a computed variable) has four connected transmission gadgets, three of which are completely colored by reusing the seven colors of the pseudolines of the von Staudt gadget, resulting again in 10 needed colors.
Such a coloring is shown in Figure~\ref{fig:gad_multiplication}.
The multiplication requires one of the factors to be represented as its reciprocal value $1/x$, which is achieved with the inversion gadget.

\begin{figure}
    \centering
    \begin{subfigure}[t]{.48\linewidth}
        \centering
        \includegraphics[page=22]{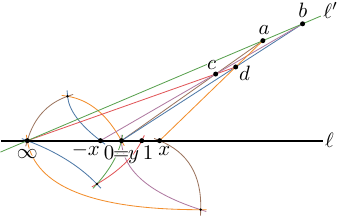}
        \subcaption{$V_i = -x$}
        \label{fig:gad_negation}
    \end{subfigure}
    \hfill
    \begin{subfigure}[t]{.48\linewidth}
        \centering
        \includegraphics[page=23]{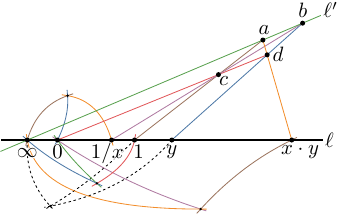}
        \subcaption{$V_i = x\cdot y$}
        \label{fig:gad_multiplication}
    \end{subfigure}
    \hfill
    \begin{subfigure}[t]{.48\linewidth}
        \centering
        \includegraphics[page=24]{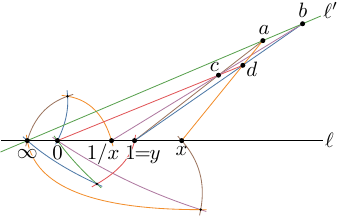}
        \subcaption{$V_i = 1/x$}
        \label{fig:gad_inversion}
    \end{subfigure}
    \hfill
    \begin{subfigure}[t]{.48\linewidth}
        \centering
        \includegraphics[page=25]{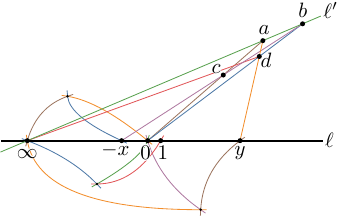}
        \subcaption{Condition $x<y$}
        \label{fig:gad_condition}
    \end{subfigure}
    \caption{Colorings of the four von Staudt gadgets for negating a variable value (a), multiplying to variable values (b), creating the reciprocal of a variable value (c) and encoding an inequality condition between two variables (d).}
    \label{fig:gad_rest}
\end{figure}

\paragraph{Inversion}
A similar relation as between the negation and addition gadget can be seen between the inversion and the multiplication gadget.
Specifically the inversion gadget, which represents a computation assigning the reciprocal value of a underlying or computed value to a new computed value.
The coloring for the inversion gadget is shown in Figure~\ref{fig:gad_inversion}.
Due to only three needed transmission gadgets, the required colors are at most 7.

\paragraph{Condition}
Finally a program in Richter-Gebert normal form can compare the value of two variables in a condition.
The von Staudt gadget for the condition (with its coloring shown in Figure~\ref{fig:gad_condition}) does not produce an output, which can be transmitted to a new computation and therefore only requires three transmission gadgets for the scale and two input variables.
The coloring can be done with at most $7$ colors.
\medskip

We have assumed that in the (negated) addition, negation, multiplication, inversion and condition gadgets exactly one transmission gadget connects from above and we can choose, which one.
This is guaranteed with Lemma~\ref{lem:connection_direction}.
Note that the pseudo-segment of the up to three transmission gadgets, which are connected from below are all covered in pairwise different colors, and therefore can cross freely.
This removes any dependence on the order of the edges incident to the vertex rectangle in the orthogonal drawing of $G$.

Based on the colorings described in the previous paragraphs and illustrated in the corresponding figures, we state the following observation.

\begin{lemma}\label{lem:staudt_gadget_coloring}
    Every von Staudt gadget including all pseudo-segments of connected transmission gadget, which intersect at least one pseudo line segment of the gadget can locally be colored with at most 10 colors if we can choose the direction from which all transmission gadgets connect.
\end{lemma}

It remains to consider the transmission gadgets.
Recall that a transmission gadget is made up of a chain of inversion gadgets.
A single inversion gadget consists of five pseudo-segments arranged as shown in Figure~\ref{fig:inversion_gadget}.
We will call the pseudolines $\ell_1$ and $\ell_2$ the \emph{end-segments} of the inversion gadget, while the other three pseudolines are referred to as the \emph{inner segments} of the gadget.
In a straight line segment realization of this gadget, the relative distances between the three points on line $\ell$ are preserved as the relative distances between the three points on $\ell'$.
In Schaefers construction, the edges between nodes (which represent von Staudt gadgets) are subdivided (an odd number of times, see also Figure~\ref{fig:subdivision_of_path}) and then every edge in these subdivided paths is replaced by an inversion gadget (as shown in Figure~\ref{fig:gad_transmission}).

\begin{figure}
    \centering
    \begin{subfigure}[t]{.15\linewidth}
        \centering
        \includegraphics[page=1,width=\linewidth]{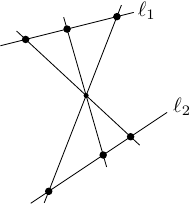}
        \subcaption{}
        \label{fig:inversion_gadget}
    \end{subfigure}
    \hfill
    \begin{subfigure}[t]{.15\linewidth}
        \centering
        \includegraphics[page=3,width=\linewidth]{figures/transmission.pdf}
        \subcaption{}
        \label{fig:subdivision_of_path}
    \end{subfigure}
    \hfill
    \begin{subfigure}[t]{.67\linewidth}
        \centering
        \includegraphics[page=4,width=\linewidth]{figures/transmission.pdf}
        \subcaption{}
        \label{fig:gad_transmission}
    \end{subfigure}
    \caption{Edges of $G$ are represented by chains of inversion gadgets (a).
    Any path in $G$ is subdivided an odd number of times and at least of length 4 (b).
    The gadgets are placed according to the vertex positions in the subdivision of $G$ (c).}
    \label{fig:enter-label}
\end{figure}

We now investigate how this sequence of inversion gadgets can be colored. In particular we want to proof the following statement.

\begin{lemma}\label{lem:independence_transmission}
    Assume we are given 10 available colors $C= \{c_1, \ldots, c_{10}\}$ and a sequence of $3$ inversion gadgets, where the inner segments of the first inversion gadget are already assigned a color each.
    Then for every sorted triple of pairwise different colors $c, c', c''$ there exists a coloring of all pseudo-segments in the gadgets, which respects the already assigned colors and the inner pseudoline-segments of the third inversion gadgets are colored with the colors $c, c', c''$.
\end{lemma}
\begin{proof}
    Assume w.l.o.g., that the inner segments of the first inversion gadget are colored in the colors $c_1, c_2, c_3$.
    Let $C' = C \setminus \{c_1, c_2, c_3, c, c', c''\}$.
    Note that $\vert C'\vert\geq 4$.
    Let $d, e, f, g$ be four pairwise different colors in $C'$.
    We can always color all end segments using $d$ and color the inner segments of the second inversion gadget using $e, f$ and $g$.
    Since no other segment in the arrangement is colored in the colors $c, c', c''$ we can use them to color the inner segments of the third inversion gadget and can do so in any desired order.
    This coloring is shown in \cref{fig:transmission_independence,fig:transmission_independence_crossing}.
\end{proof}

Note that the proof does not assume that any of the inner segments of the inversion gadgets are crossing-free.
Therefore this lemma holds true even if the inner segments of two consecutive inversion gadgets cross as shown in Figure~\ref{fig:transmission_independence_crossing}.

\begin{figure}
    \centering
    \begin{subfigure}[t]{.45\linewidth}
        \centering
        \includegraphics[page=5]{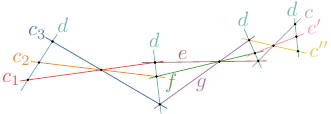}
        \subcaption{}
        \label{fig:transmission_independence}
    \end{subfigure}
    \hfill
    \begin{subfigure}[t]{.25\linewidth}
        \centering
        \includegraphics[page=6]{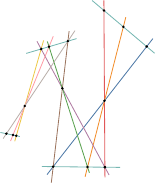}
        \subcaption{}
        \label{fig:transmission_independence_crossing}
    \end{subfigure}
    \hfill
    \begin{subfigure}[t]{.25\linewidth}
        \centering
        \includegraphics[page=7]{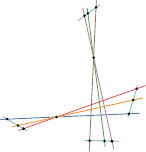}
        \subcaption{}
        \label{fig:transmission_independence_intersect}
    \end{subfigure}
    \caption{The coloring of the inner segments of the third inversion gadget in the sequence of inversion gadgets (a) is independent of the fixed coloring of the three inner segments of the first inversion gadget with colors $c_1, c_2$ and $c_3$.
    This is also true if the inner segments of consecutive inversion gadgets intersect (b).
    Any crossing of two inversion gadgets only involves at most $6$ pseudo-segments (c).}
    \label{fig:transmission_2}
\end{figure}

Since the paths between gadgets in Schaefers construction are at least of length four, we can state that the coloring of the middle segment of the inversion gadgets connected directly to any von Staudt gadget can be done exactly as described above and shown in Figures~\ref{fig:gad_assignment},~\ref{fig:gad_addition} and~\ref{fig:gad_rest}.

Finally since the placement of all gadgets is done according to a 1-planar drawing, the inner segments of two different inversion gadgets of two different sequences of inversion gadgets representing two edges might cross.
However since the drawing is 1-planar it is sufficient to subdivide every edge of the drawing twice to guarantee that if two inversion gadgets $A$ and $B$ of the same sequence are both involved in a crossing with inversion gadgets of other sequences, then $A$ and $B$ can not be consecutive in the sequence.
Therefore we can use Lemma~\ref{lem:independence_transmission} to consider the coloring of the inner segments of two crossing inversion gadgets independently of any other crossing.
It is then obvious that this can be done with 10 colors, since only 6 pseudolines cross (see Figure~\ref{fig:transmission_independence_intersect}).

As a result of this we can state the following lemma.

\begin{lemma}\label{lem:transmission_coloring}
    Every inversion gadget in Schaefers construction can be colored with at most 10 colors even if all inner segments of inversion gadgets which intersect a von Staudt gadget have fixed colors.
\end{lemma}

We are now prepared to combine the so far stated results into the following theorem.

\begin{theorem}
    Schaefers non-uniform pseudo-segment arrangement construction can be adapted, s.t., all segments can be colored using at most 10 colors.
\end{theorem}
\begin{proof}
    By adapting the graph to guarantee low degree using Lemma~\ref{lem:degree_replacement} and creating a 1-planar 9-bend orthogonal drawing as described in Lemma~\ref{lem:connection_direction}. Then we can obtain a straight line drawing (following the steps in Schaefers reduction), which can be used as the basis of the placement of von Staudt and transmission gadgets.
    In the resulting set of pseudo-segments, we can color all von Staudt gadgets according to the observations in Lemma~\ref{lem:staudt_gadget_coloring} using at most 10 colors.
    This also already colors all inner segments of inversion gadgets, which intersect the von Staudt gadgets.
    Since every edge in the drawing we used as a basis for the arrangement is crossed by at most one other edge and no two edges, which are involved in a crossing share an endpoint, we can extend the coloring of all von Staudt gadgets and the immediately connected transmission gadgets as obtained in Lemma~\ref{lem:transmission_coloring} to a coloring of the entire arrangement by coloring the transmission gadgets according to Lemma~\ref{lem:transmission_coloring}. This new coloring still uses at most the already used 10 colors of the gadget coloring.
\end{proof}

\subsection{Lowering the Impact of Making the Arrangement Uniform}

In this section, we will show that the pseudo-segment arrangement from Schaefer's construction has chromatic number $\le\chromaticnumber$ by analyzing how the needed number of colors for the pseudo-segment arrangement changes in the last step of Schaefers proof, i.e., the dual technique by Las Vergnas~\cite{Vergnas86}.

\corColoring*
\label{cor:coloring*}
\begin{proof}
    Schaefer (and therefore our) constructed pseudo-segment arrangement is not uniform, i.e., it can happen that three or more segments cross in a single point.
    However, the machinery only works with uniform pseudo-segment arrangements, so he shows that his construction yields a so-called constructible arrangement and uses the dual of a technique by Las Vergnas~\cite{Vergnas86} to make the arrangement uniform.
    To this end, some pseudo-segments are replaced by 2 or 4 other pseudo-segments as shown in \cref{fig:lasvergnas}, which increases the number of crossings per segments
    by factor 4.
    However, it is straight-forward to see that these replacements increase the chromatic number by at most factor 3.
    When we apply the same technique to the adapted construction the chromatic number increases from 10 to \chromaticnumber.

    \begin{figure}
        \centering
        \includegraphics{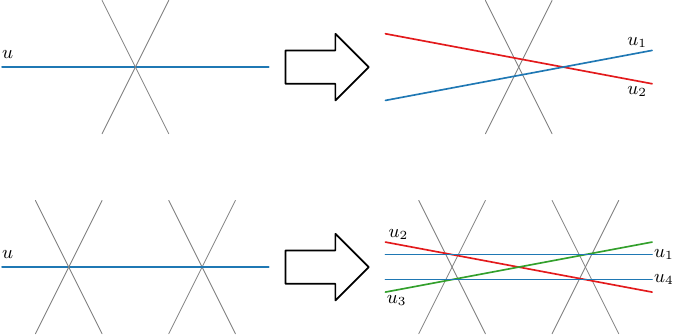}
        \caption{Making a pseudo-segment arrangement uniform, increasing the chromatic number by factor 3.}
        \label{fig:lasvergnas}
    \end{figure}

    Let $\chi'$ be a 10-coloring of the pseudo-segment arrangement we created based on our adapted drawing of $G'$ 
    where every vertex $u$ has a color $c'(u)\in\{c_1,\ldots,c_{19}\}$. We obtain a \chromaticnumber-coloring $\chi$ of $G$ where every vertex $v$ has a color $c(v)\in\{c_1^1,c_1^2,c_1^3,\ldots,c_{10}^1,c_{10}^2,c_{10}^3\}$ as follows. Let $u$ be a vertex in $G'$ with color $c'(u)=c_i$. If $u$ is not replaced in the construction, we set $c(u)=c_i^1$. If $u$ is replaced by two vertices $u_1,u_2$, we set $c(u_1)=c_i^1$ and $c(u_2)=c_i^2$. If $u$ is replaced by four vertices $u_1,u_2,u_3,u_4$, observe that $u_1$ and $u_2$ are not adjacent; see \Cref{fig:lasvergnas}(bottom). Hence, we can set $c(u_1)=c(u_4)=c_i^1$, $c(u_2)=c_i^2$, and $c(u_3)=c_i^3$. 
\end{proof}

\section{Conclusion}

The two most important open question related to our work are to resolve the complexity also for \textit{simple} graphs 
and for more realistic bounds on the geometric thickness, in particular for geometric thickness two. To this end, note that simple graphs are more widely studied in the graph drawing literature, are arguably more natural as mathematical objects and also naturally appear in applications. Hence, it would be interesting to know the precise complexity for simple graphs as well. We believe that our result on multigraphs may become a precursor to understanding the geometric thickness of simple graphs. 

In order to generalize our result to simple graphs, it might be desirable to construct a gadget graph in a way that provides us with some control
both about the geometric embedding as well as how the edges are colored. 
Ideally, we would like to construct a graph in a way that any coloring realizing its geometric thickness leads to the vertices being connected in all colors. 

\begin{question}
\label{q:connected}
 Given $t \in \N$, does there always exist a graph with geometric thickness $t$ such that
    any $t$-colored drawing of $G$ realizing its geometric thickness is connected in all $t$ colors?
\end{question}

Such a connected construction seems elusive and might not even be possible.
It might be worth noting that the minimum number of common edges between any two triangulations
is known to be five and eight~\cite{hutchinson1999representations}.

Regarding the chromatic number of the considered pseudo-segment arrangements, an intriguing question related to our reduction is the following. 
\begin{question}
    \label{q:2coloring} Is
    \segmentStretchability \ER-hard even for pseudo-segment arrangements 
    of chromatic number two?
\end{question}

In this direction,  we already showed that modifications of Schaefer's reduction lower the current best bound on the chromatic number for \ER-hardness from 73 to \chromaticnumber in \Cref{sec:chromaticnumber}. 
We believe that considerably new ideas are needed in order to improve upon the chromatic number and eventually show \Cref{q:2coloring}.
The significance of resolving \Cref{q:2coloring} is that it could be a major step to
show \ER-hardness of many graph drawing problems with bounded parameter.
Note that given a positive answer to \Cref{q:2coloring}, our proof in \Cref{sec:geothickness} implies directly that \geo is \ER-complete already for geometric thickness two, and \sge is \ER-complete for three graphs forming an empty sunflower.

Further, we remark that given positive answers to both \Cref{q:connected,q:2coloring} our reduction implies that  \geo is \ER-complete already for \textit{simple} input graphs and thickness two. Namely, we could replace each edge of multiplicity $t$ in our construction with the gadget graph \cref{q:connected} asks for.

Finally, it is important to study practical solutions for computing graph drawings of low geometric thickness. To this end, note that we know very little about the geometric thickness of even very simple graph classes, e.g., we do not know the geometric thickness of the complete graph $K_n$, for all values of $n$~\cite{DEH00}. Hence it is not too surprising that we aware of efficient embedding algorithms only for  specific settings,  e.g. for maximum degree four  graphs there is an algorithm achieving geometric thickness two~\cite{DuncanEK04lowdeg}.
In more general settings, we are only aware of two inefficient ``all purpose'' approaches to check the geometric thickness of a graph: encoding as an ETR sentence or checking all possible order types with a SAT solver. For the former approach, algorithms to decide an ETR sentence are very slow and might 
not be feasible for more than 5 or 6 vertices. For the latter approach, we may use databases of all order types~\cite{AichholzerAK2002otdatabase,OEIS_A006247} for  $n \leq 11$; however note that  there are over 2~billion different order types with 11 points. Therefore, a scalable general approach is of great interest.

\paragraph{Acknowledgements.}
All authors would like to express their gratitude towards
the 18th European Research Week on Geometric Graphs which was held in Alcal\'{a} de Henares, Spain, in September 2023.
The workshop participants, the city, the university, and the organizers provided a unique experience that made this research possible. 
A preliminary version of this work has been presented at the Latin American Theoretical INformatics Symposium 2024 in Puerto Varas, Chile. We thank the reviewers of that preliminary version for their valuable feedback.
Finally, we thank Tony Huynh for pointing us at relevant literature about graphs with common edges.

\paragraph{Funding.}

T. M. is generously supported by the Netherlands Organisation for Scientific Research (NWO) under project no. VI.Vidi.213.150. 
I. P. is a Serra H\'unter Fellow. Partially supported by grant 2021UPC-MS-67392 funded by the Spanish Ministry of Universities and the European Union (NextGenerationEU) and by grant PID2019-104129GB-I00 funded by MICIU/AEI/10.13039/501100011033.
S. T. has been funded by the Vienna Science and Technology Fund (WWTF) [10.47379/ICT19035] and by the Dutch Research Council (NWO) through Gravitation-grant NETWORKS-024.002.003.

\newpage
\bibliographystyle{alphaurl}
\bibliography{ETR2,lib}

\end{document}